\documentclass[superscriptaddress,aps,twocolumn,babel]{revtex4-1}
\usepackage{amsmath}
\usepackage{amssymb}
\usepackage{siunitx}
\usepackage{amsthm}
\usepackage{physics}
\usepackage{verbatim}
\usepackage{graphicx}
\usepackage{appendix}

\usepackage{times}	
\usepackage{tikz}
\usepackage{soul}
\newcommand{\be}{\begin{equation}}
\newcommand{\ee}{\end{equation}}
\newcommand{\ba}{\begin{array}}
\newcommand{\ea}{\end{array}}
\newcommand{\bea}{\begin{eqnarray}}
\newcommand{\eea}{\end{eqnarray}}

\usepackage[colorlinks = true]{hyperref}
\usepackage{xcolor}
\definecolor{darkred}  {rgb}{0.5,0,0}
\definecolor{darkblue} {rgb}{0,0,0.5}
\definecolor{darkgreen}{rgb}{0,0.5,0}
\hypersetup{
  urlcolor   = blue,         
  linkcolor  = darkblue,     
  citecolor  = darkgreen,    
  filecolor  = darkred       
}

\newcommand{\ra}{\rangle}
\newcommand{\la}{\langle}

\newcommand{\calL}{{\cal L }}

\newcommand{\calF}{{\cal F }}

\newcommand{\calS}{{\cal S }}

\newcommand{\calP}{{\cal P }}

\makeatletter
\def\blfootnote{\xdef\@thefnmark{}\@footnotetext}
\makeatother

\newtheorem{prop}{Proposition}
\newtheorem{lemma}{Lemma}

\newtheorem{theorem}{Theorem}
\newcommand{\ii}{{i}}

\usepackage{qcircuit}
\begin{document}

\title{Superfast encodings for fermionic quantum simulation}

\author{Kanav Setia}
\affiliation{Department of Physics and Astronomy, Dartmouth College, Hanover, NH 03755}
\email{kanav.setia.gr@dartmouth.edu}
\affiliation{IBM T.J. Watson Research Center, Yorktown Heights, NY 10598, US}

\author{Sergey Bravyi} 
\affiliation{IBM T.J. Watson Research Center, Yorktown Heights, NY 10598, US}

\author{Antonio Mezzacapo} 
\affiliation{IBM T.J. Watson Research Center, Yorktown Heights, NY 10598, US}

\author{James D. Whitfield} 
\affiliation{Department of Physics and Astronomy, Dartmouth College, Hanover, NH 03755}

\date{\today}
\begin{abstract}
Simulation of fermionic many-body systems on a quantum computer requires
a suitable encoding of fermionic degrees of freedom into qubits. 
Here we revisit the Superfast Encoding introduced by Kitaev and one of the authors.
This encoding maps a target fermionic Hamiltonian with two-body interactions
on a  graph of degree $d$ to a qubit simulator Hamiltonian composed of Pauli
operators of weight $O(d)$. A system of $m$ fermi modes gets
mapped to $n=O(md)$ qubits. 
We propose Generalized Superfast Encodings (GSE)
which require the same number of qubits as the original one  but have
more favorable properties. First, we describe a GSE such that
the corresponding quantum code
corrects any single-qubit error provided that the
interaction graph has degree $d\ge 6$.
In contrast,  we prove that the original Superfast Encoding 
lacks the error correction property for $d\le 6$.
Secondly, we describe a GSE that  reduces the Pauli weight 
of the simulator Hamiltonian from $O(d)$ to $O(\log{d})$.
The robustness against errors and a simplified structure of the simulator Hamiltonian
offered by GSEs can make simulation of fermionic systems within the reach of near-term quantum devices.
As an example, we apply the new encoding  to  the  fermionic Hubbard model on a 2D lattice. 

\end{abstract}

\maketitle

Quantum error correction is  a vital milestone  on the path
towards scalable quantum computing.
It  enables an arbitrarily long reliable computation with noisy qubits
and imperfect gates, provided that the noise level is below
a constant threshold value, which is close to what can be
achieved in the latest experiments~\cite{barends2014superconducting,corcoles2015}.
Unfortunately, realizing
a computationally universal set of logical gates in a fully fault-tolerant fashion
requires a significant overhead which may be prohibitive for
near-term quantum devices. This has lead several groups
to consider a less expensive option known as error 
mitigation~\cite{temme2017error,li2017efficient,li2018practical,otten2018recovering,bonet2018low}.
Error mitigation schemes are usually tailored to a specific quantum algorithm such as
adiabatic quantum computation~\cite{jordan2006error} or
variational optimization~\cite{temme2017error,li2017efficient}.
Although the proposed error mitigation schemes introduce less overhead and can extend
the range of applications for the available quantum hardware~\cite{kandala2018extending}, 
they are not truly scalable and do not offer full
fault-tolerance.

Of particular interest for practical applications 
are error mitigation schemes for 
quantum simulation of  fermionic systems --
a fundamental  problem emerging in the quantum chemistry and 
material science. 
All quantum algorithms for simulation of fermionic systems rely on 
a suitable encoding of fermionic degrees of freedom into qubits.
Notable examples are  the Jordan-Wigner transformation~\cite{Jordan1928}, 
the Verstraete-Cirac mapping~\cite{Verstraete2005},
Fenwick trees~\cite{Bravyi2000,Havlicek2017},
and the parity mapping~\cite{Seeley2012}, see also~\cite{Ball2005,bravyi2017tapering,steudtner2017lowering}.
Such encodings map a target Hamiltonian $H$ describing $m$ fermionic modes to
a simulator Hamiltonian $\tilde{H}$ describing $n$ qubits such that
$VH=\tilde{H}V$ for a suitable unitary map (isometry) $V$.
This ensures that $H$ and $\tilde{H}$ are unitarily equivalent
if one restricts $\tilde{H}$ onto the subspace spanned by encoded states $V|\psi\ra$.

A natural question is whether the encodings proposed for fermionic simulations can also be used for error correction or mitigation. Here we give the affirmative answer, for a generalized version of the Superfast Encoding proposed in Ref.~\cite{Bravyi2000}.
We consider a  system of $m$ fermi modes that live at vertices of some
graph with the maximum vertex degree $d\ll m$.
Edges of the graph represent two-mode interactions 
in the target Hamiltonian $H$.
A distinctive feature of the superfast encodings is that the simulator Hamiltonian
$\tilde{H}$ includes only few-qubit interactions described by Pauli
operators of  weight $O(d)$.  
The encoding requires $n=O(md)$ qubits.
For comparison, the Jordan-Wigner and the Fenwick-tree type of encodings
require $n=m$ qubits and  produce a simulator Hamiltonian with Pauli weight $\Omega(m)$
and $\Omega(\log{m})$ respectively.

Here we propose Generalized Superfast Encodings (GSE) improving the
original  Superfast Encoding in two respects. 
First, we describe a GSE
such that the corresponding  quantum code 
corrects any single-qubit error under mild technical assumptions
about the fermionic interaction graph.
Namely, we assume that the  graph is $3$-connected~\footnote{Recall that a graph
is called $3$-connected if it remains connected after removal
of any pair of vertices.} and 
has vertex degree $d\ge 6$. In contrast, we prove that the original Superfast Encoding  
lacks the error correction property for $d\le 6$. 
The GSE  requires the same number
of qubits  as the original encoding, so the extra error correction
capability comes at no extra qubit cost. Additionally, the GSE produces a simpler local simulator Hamiltonian, with Pauli weights reduced by a factor 2 with respect to the original encoding.

 Secondly, we describe a GSE that 
produces a simulator Hamiltonian with the Pauli weight $O(\log{d})$,
as opposed to the Pauli weight $O(d)$ in the original Superfast Encoding.
Both encodings use the same number of qubits. 
This achieves a significant reduction of the Pauli weight compared to all
previously known encodings in the regime when $d\ll m$. Note that this is
naturally the case of physical systems, where the interactions
have some degree of locality independent of the system size $m$.

We expect that the proposed GSEs will 
find practical applications in simulation of medium-size fermionic systems
 with aim at correcting single-qubit errors that occur in noisy devices. 
Furthremore, reducing Pauli weight of the simulator Hamiltonian is vital for
mitigating readout errors 
in variational quantum algorithms~\cite{peruzzo2014variational,Kandala2017}.

The paper is organized as follows.
We first define the relevant fermionic Hamiltonians 
and review the Superfast Encoding of Ref.~\cite{Bravyi2000}.
Then we introduce GSEs and show that they 
can exponentially reduce the Pauli weight of the simulator Hamiltonian.
We prove that the original Superfast Encoding lacks the error correction property on low-degree 
graphs ($d \leq 6$).
In contrast, we demonstrate that GSEs correct all single-qubit errors
for any $3$-connected  interaction graph with vertex degree at least $6$.
Finally, we elucidate a practical use of GSEs 
by applying them to a Hubbard model on a square lattice. 

{\it Superfast Encoding.}
We start by summarizing the encoding proposed in Ref.~\cite{Bravyi2000}.
Consider a system of $m$ fermionic modes
such that each mode can be either empty or occupied by a fermionic particle. 
Let $a_i^\dag$ and $a_i$ be the particle creation and annihilation operators
for the $i$-th mode.  
They obey the canonical commutation rules
\[
a_i a_j + a_j a_i =0 \quad \mbox{and} \quad a_i a_j^\dag + a_j^\dag a_i =\delta_{i,j}I.
\]
Let $N=\sum_{j=1}^N a_j^\dag a_j$ be the particle number operator.
A fermionic operator $H$ is called {\em even} if it preserves the number
of particles modulo two, that is, $[H,(-1)^N]=0$.
All physical observables are known to be described by even operators.
Let $\calF$ be the algebra of all even operators.

Assume that each mode $i$ can interact only with a few other
modes $j$ that are nearest-neighbors of $i$ on some  graph
$G=(V,E)$ with a set of vertices $V=\{1,2,\ldots,m\}$
and a set of edges $E$. Such system is described by a Hamiltonian
\be
\label{Htarget}
H=\sum_{(i,j)\in E} H_{i,j},
\ee
where $H_{i,j}\in \calF$ acts non-trivially only on the pair of modes $i,j$.
Below we assume that $G$ is a connected graph. 

To define the Superfast Encoding it is convenient to rewrite $H$ in terms
of Majorana operators
\be
c_{2j}= a_{j}+a^{\dagger}_{j} \quad \mbox{and} \quad 
c_{2j+1}= -\ii(a_{j}-a^{\dagger}_{j}).
\ee
These operators are Hermitian and satisfy
\begin{align}
c_{j}c_{k}+c_{k}c_{j}=2\delta_{jk} I \label{eq:Maj_modes}
\end{align}
The algebra of even operators $\calF$ has a set of generators
\begin{align}
	B_j&=-\ii c_{2j}c_{2j+1} \quad \text{for each vertex $j\in V$,}\label{eq:B_i}\\
	A_{jk}&=-\ii c_{2j}c_{2k} \quad \text{for each edge }(j,k) \in E.	\label{eq:A_ij}
\end{align}
For example, fermionic operators 
describing hopping, external field, and a two-body repulsion
can be written as
\[
a_j^\dag a_k + a_k^\dag a_j = (-i/2)A_{j,k} (B_j - B_k),
\]
\[
a_j^\dag a_j = (1/2)(I-B_j), \quad
a_j^\dag a_j a_k^\dag a_k = (1/4)(I-B_j)(I-B_k).
\]
Any parity-preserving fermionic operator belongs to the subalgebra generated by $A_{j,k}$, $B_j$. An explicit derivation of two-body quantum chemistry and superconductivity interactions can be found in~\cite{Setia2017}, and in Appendix~\ref{app:ferm_op_table}.

The operators
$A_{i,j}$ and $B_i$ obey commutation rules
\begin{align}
 & B^{\dagger}_{i}=B_{i},\qquad \qquad \quad \quad  \qquad A^{\dagger}_{ij}=A_{ij}, \label{CR1}\\
& B_{i}^{2}=1,\qquad \qquad \quad \qquad \qquad A_{ij}^{2}=1, \\
 &B_{i}B_{j}=B_{j}B_{i}, \qquad \quad \qquad \quad  A_{ij}=-A_{ji} \\
 &\quad \quad A_{ij}B_{k}=(-1)^{\delta_{ik}+\delta_{jk}}B_{k}A_{ij}, \\
 &\quad A_{ij}A_{kl}=(-1)^{\delta_{ik}+\delta_{il}+\delta_{jk}+\delta_{jl}}A_{kl}A_{ij} \label{CR2} \\
 &\quad i^s A_{\zeta(0),\zeta(1)} A_{\zeta(1),\zeta(2)} \cdots A_{\zeta(s-1),\zeta(0)}=I.  \label{eq:stab_relation}
 \end{align}
In the last equation $\zeta$ is any closed loop of length $s$
in the graph $G$
that consists of vertices $\zeta(0),\zeta(1),\ldots,\zeta(s)=\zeta(0)\in V$.
Following Ref.~\cite{Bravyi2000} we shall impose one extra rule
\be
\label{CR3}
 \prod_{i\in V} B_i = I.
\ee
Note that $\prod_{i\in V} B_i=(-1)^N$. Thus Eq.~(\ref{CR3})
corresponds to restricting the Fock space of $m$ fermi modes to 
the subspace with even number of particles. Note that all elements of the 
algebra $\calF$ preserve this subspace.

To define the simulator Hamiltonian $\tilde{H}$
let us place a qubit at every edge of the graph $G$.
The total number of qubits is 
\be
\label{n}
n=|E|=(1/2) \sum_{i\in V} d(i)
\ee
where $d(i)$ is the degree of a vertex $i$.
Let $X_{ij}$, $Y_{ij}$, and $Z_{ij}$ be the Pauli operators acting on the
edge $(i,j)\in E$. We shall assume that edges incident to each vertex $i$
are labeled by integers $1,\ldots,d(i)$. 
The corresponding ordering of edges incident to $i$
will be denoted $<_i$.
We shall also  assume that every edge $(i,j)$
is oriented. Define $\epsilon_{i,j}=1$ if $i$ is the head and $\epsilon_{i,j}=-1$
if $i$ is the tail of the edge $(i,j)$.
Qubit counterparts of the operators $B_j$ and  $A_{j,k}$  are defined as 
\begin{equation}
\tilde B_j=\prod_{k: (j,k)\in E} Z_{jk} \label{tilde_B_i},
\end{equation} 
\begin{equation}
\tilde A_{jk}=\epsilon_{jk} X_{jk}\prod_{p\, :\, (j,p)<_j(j,k)}\; Z_{jp} \prod_{q\, : (k,q)<_k(k,j)} \; Z_{kq}\label{tilde_A_ij}.
\end{equation}
It can be checked that these operators satisfy commutation rules
analogous to  Eqs.~(\ref{CR1}-\ref{CR2}) and Eq.~(\ref{CR3}). 
However, the  rule Eq. ~(\ref{eq:stab_relation}) does not hold on the full Hilbert space of $n$ qubits.
This rule can be satisfied by restricting the operators $\tilde{A}_{i,j}$ and $\tilde{B}_i$
on a suitable subspace.
For each closed loop $\zeta$ as above  define
a loop operator
\begin{align}
\label{LoopOperator}
\tilde{A}(\zeta)\equiv  i^s \tilde{A}_{\zeta(0),\zeta(1)} \tilde{A}_{\zeta(1),\zeta(2)} \cdots 
&\tilde{A}_{\zeta(s-1),\zeta(0)}.
\end{align}
Recall that $s$ is the length of $\zeta$.
It can be readily checked that $\tilde{A}(\zeta)$ commutes with 
all operators  $\tilde{A}_{i,j}$ and $\tilde{B}_i$. Furthermore, 
loop operators  pairwise commute.
Let $\calS$ be the abelian group generated by the loop operators $\tilde{A}(\zeta)$.
In Appendix~\ref{app:loop} we show that  $-I\notin \calS$. Thus $\calS$ can be viewed
as a stabilizer group of a quantum code with the codespace
\begin{align}
\label{logical}
\calL = \{  \ket{\psi}\, : \, \tilde{A}(\zeta)|\psi\ra=|\psi\ra \quad \mbox{for all loops $\zeta$} \}.
\end{align}			
The number of independent stabilizers coincides with the number of 
independent loops in the graph which is known to be $s=|E|-|V|+1=n-m+1$.
It follows that the code $\calS$ encodes $k=n-s=m-1$ logical qubits
into $n$ physical qubits, that is, $\mathrm{dim}(\calL)=2^{m-1}$.
The codespace $\calL$ can be identified with the even-parity subspace
of the fermionic Fock space. Furthermore, the restrictions of
qubit operators $\tilde{A}_{i,j}$ and $\tilde{B}_i$ onto $\calL$
can be identified with the fermionic operators
$A_{i,j}$ and $B_i$ restricted onto the even-parity subspace.
We can now define a simulator Hamiltonian $\tilde{H}$ by 
replacing the operators $A_{i,j}$ and $B_i$ in the expansion of each term $H_{i,j}$
by their qubit counterparts $\tilde{A}_{i,j}$ and $\tilde{B}_i$.
One can easily check that $\tilde{H}$ is composed of Pauli operators
of weight at most $2d$.

{\it Generalized Superfast Encodings (GSE). }
Consider the target Hamiltonian Eq.~(\ref{Htarget}).
Below we assume that the interaction graph $G=(V,E)$ is connected and has only even
degree vertices. Let us place
$d(i)/2$ qubits at each vertex $i$. 
The total number of qubits $n$ is given by Eq.~(\ref{n}).
Let $\calP_i$ be the group of Pauli operators acting on the  qubits
located at a vertex $i\in V$
tensored with the identity on the remaining qubits. 
A GSE
is defined in terms of {\em local Majorana modes} 
\be
\label{LMM}
\gamma_{i,1},\gamma_{i,2},\ldots,\gamma_{i,d(i)} \in \calP_i.
\ee
By definition, $\gamma_{i,p}$ is a  Pauli operator acting non-trivially 
only on the qubits located at the vertex $i$. 
We require that the operators $\gamma_{i,p}$ generate the full Pauli group
$\calP_i$ and 
obey the usual Majorana commutation rules 
\be
\label{localMajorana}
\gamma_{i,p}^\dag = \gamma_{i,p}, \quad 
\gamma_{i,p} \gamma_{i,q} + \gamma_{i,q} \gamma_{i,p} = 2\delta_{p,q} I
\ee
for all $i\in V$ and $1\le p,q\le d(i)$.  
Otherwise, $\gamma_{i,p}$ can be completely arbitrarily.
Hence a GSE is non-unique.
Note that  local Majorana modes located at different vertices commute with each other
because they act on disjoint subsets of qubits. 
Assume that the
neighbors of each vertex $i$ are labeled by integers $1,\ldots,d(i)$
and denote the $p$-th neighbor  of $i$ as $N(i,p)$.
Define qubit counterparts of the operators $B_j$ and $A_{j,k}$ as 
\be
\label{tildeBgen}
\tilde{B}_j =(-i)^{d(j)/2} \gamma_{j,1} \gamma_{j,2} \cdots \gamma_{j,d(j)}
\ee
and
\be
\label{tildeAgen}
\tilde{A}_{j,k}=\epsilon_{j,k} \gamma_{j,p} \gamma_{k,q}
\ee
where the integers $p,q$ must satisfy
\[
k=N(j,p) \quad \mbox{and} \quad j=N(k,q).
\]
In other words, $k$ is the $p$-th neighbor of $j$ while $j$ is the $q$-th neighbor of $k$,
see Fig.~\ref{fig:LMM}.
\begin{figure}[t!]
	\includegraphics[width=6cm]{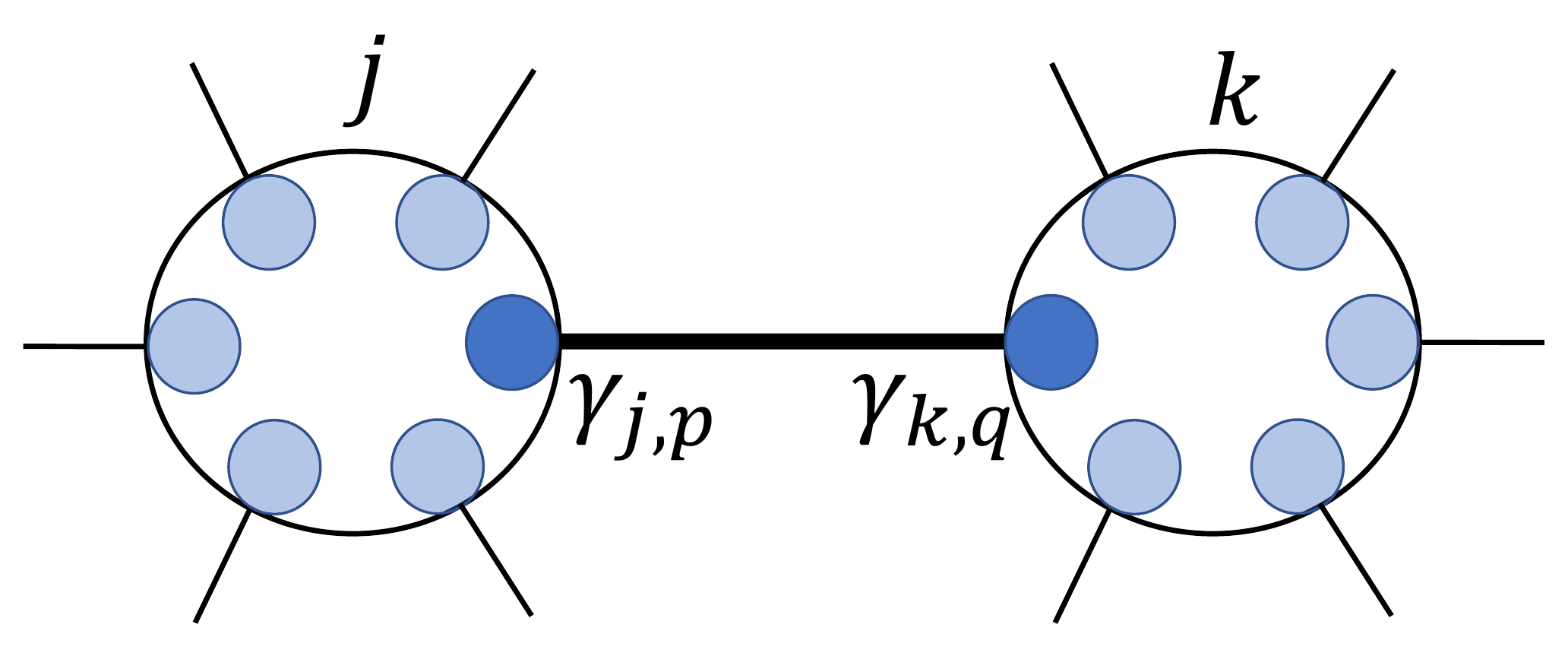}
	\caption{Local Majorana modes for nearest-neighbor vertices $j,k$.
	We define $\tilde{A}_{j,k}=\epsilon_{j,k} \gamma_{j,p}\gamma_{k,q}$,
	where $\epsilon_{j,k}=\pm 1$ is the edge orientation. The operator
	$\tilde{B}_j$ is proportional to the product of all  local Majorana
	modes $\gamma_{j,p}$ located at the vertex $j$.}
    \label{fig:LMM}
\end{figure}
One can check that $\tilde{B}_i$ and $\tilde{A}_{i,j}$ obey the commutation rules
analogous to  Eqs.~(\ref{CR1}-\ref{CR2}).
To satisfy the remaining rules Eqs.~(\ref{eq:stab_relation},\ref{CR3}) 
consider an abelian group $\calS$ generated by the loop
operators $\tilde{A}(\zeta)$ 
and a codespace $\calL$ stabilized by $\calS$ as defined
in Eqs.~(\ref{LoopOperator},\ref{logical}).
In Appendix~\ref{app:loop} we show that $-I\notin \calS$ and 
the codespace  $\calL$ has dimension $2^{m-1}$.  
Recall that $G$  is assumed to be a connected even-degree graph.
It is a well-known fact any such graph has an Eulerian cycle -- a closed loop $\zeta$
that uses every edge of $G$ exactly once. 
The corresponding loop operator
has the form 
$\tilde{A}(\zeta)=\pm \prod_{i\in V} \tilde{B}_i$, where the sign
depends on the choice of edge  orientations $\epsilon_{j,k}$.
Thus the product of all operators $\tilde{B}_i$ is in the stabilizer group
$\calS$
for a suitable choice of $\epsilon_{j,k}$.
We conclude that  the restrictions of operators $\tilde{A}_{j,k}$ and $\tilde{B}_j$
onto $\calL$ satisfy the same commutation rules as the respective
fermionic operators $A_{j,k}$ and $B_j$ restricted to the even-parity subspace
of the Fock space.
We can now define a simulator Hamiltonian $\tilde{H}$ by 
replacing the operators $A_{i,j}$ and $B_i$ in the expansion of each term $H_{i,j}$
by their qubit counterparts $\tilde{A}_{i,j}$ and $\tilde{B}_i$.

Next let us describe a specific choice of the local Majorana  modes
$\gamma_{i,p}$ that results in a simulator Hamiltonian $\tilde{H}$
with the Pauli weight $O(\log{d})$.
Consider a vertex $i\in V$ and a system
of $d(i)$ Majorana modes $\gamma_1,\ldots,\gamma_{d(i)}$.
Let $\tilde{\gamma}_p$ be the Fenwick tree encoding~\cite{Bravyi2000,Havlicek2017}
of $\gamma_p$. As was shown in Ref.~\cite{Havlicek2017},
$\tilde{\gamma}_p$ is a Pauli operator of weight at most
$\lceil \log_2{d(i)}\rceil$.
Choose $\gamma_{i,p}$ as a tensor product of $\tilde{\gamma}_p$
on the vertex $i$ and the identity operator on all other vertices.
Then  $\tilde{A}_{i,j}$ has Pauli weight 
at most $2\lceil \log_2{d}\rceil$, see Eq.~(\ref{tildeAgen}).
The Fenwick tree 
encoding maps the fermionic parity operator $\gamma_1 \gamma_2 \cdots \gamma_{d(i)}$ to
a single-qubit Pauli $Z$, see~\cite{Havlicek2017}.
Hence $\tilde{B}_i$ has Pauli weight $1$, see Eq.~(\ref{tildeBgen}).
We conclude that $\tilde{H}$ has Pauli weight at most $2\lceil \log_2{d}\rceil$.

{\it Lack of error correction in the Superfast Encoding. }
Let us first discuss error correcting properties of the original Superfast Encoding.
Recall that a Pauli operator $P$ is said to be 
a logical operator for a quantum code with a stabilizer group $\calS$ if
$P$ commutes with all elements of $\calS$ 
and the restriction of $P$ onto the logical subspace $\calL$ is a
non-trivial operator. 
A code is said to correct single-qubit errors if any logical operator $P$
has weight at least three (i.e. $P$ acts nontrivially on at least three qubits).
Let us show now that 
the stabilizer code defined through the Superfast Encoding
fails to correct all single-qubit errors (regardless of how one chooses
edge ordering).

Suppose first that the interaction graph $G$ has a vertex $i$ with degree $d(i)\le 4$.
Note that  $\tilde{A}_{i,j}$ and $\tilde{A}_{i,j}\tilde{B}_j$  are logical operators
of the code $\calS$ for any $(i,j)\in E$.
We claim that at least one of these logical operators  has weight $1$ or $2$.
Indeed,  let
$(i,j)$ be the first edge incident to $i$ according to the ordering $<_i$.
Let $e(1),\ldots,e(d)$ be the edges incident to $j$ listed according to the ordering $<_j$.
Here $d\equiv d(i)\le 4$.
Suppose $(i,j)$ is the $p$-th edge incident to $j$,
that is, $(i,j)=e(p)$.
Eq.~(\ref{tilde_A_ij}) gives
\[
\tilde{A}_{i,j}=\epsilon_{i,j} X_{ij}Z_{e(1)}\cdots Z_{e(p-1)}.
\]
If $p\le 2$ then $\tilde{A}_{i,j}$ has weight $1$ or $2$.
Otherwise,
if $p\ge 3$, then $\tilde{A}_{i,j}\tilde{B}_j \sim Y_{ij}Z_{e(p+1)} \cdots Z_{e(d)}$ 
has weight $1$ or $2$. Thus
the stabilizer code $\calS$ fails to correct all single-qubit errors
regardless of how one orders the edges. 
In Appendix~\ref{app:lackEC} we extend this argument to more general graphs
and prove the following.
\begin{lemma}
\label{lemma:lackEC}
Suppose the interaction graph $G$ 
has degree $d$ for each vertex $i$.
If $d\le 6$ then 
the Superfast Encoding does not correct all single-qubit errors.
\end{lemma}
In spite of this negative result, in Appendix~\ref{app:Hubbard} we show that in certain special cases
the error correction property can be achieved by introducing ancillary fermi modes.

{\it Error Correction in the Generalized Superfast Encoding. }
Here we describe a GSE that can correct all single-qubit errors.
Below we consider arbitrary interaction graphs $G$.
We allow $G$ to have multiple edges, that is,
some pairs of vertices can be connected by more than one edge.
Recall that  a graph is called $3$-connected if it remains connected after removal
of any subset of less than three vertices.  Our main result is the following.
\begin{theorem}
\label{thm:main}
Suppose the interaction graph $G$ is $3$-connected and each vertex $i$
has even degree $d(i)\ge 6$. Suppose any pair of vertices is connected by
at most two edges.
Then the Generalized Superfast Encoding corrects all single-qubit errors.
\end{theorem}
Note that the GSE  lacks the error correction property
if $d(i)<6$ for some vertex $i$. Indeed, in this case $\tilde{B}_i$
is a logical operator with weight at most $2$ (since the vertex $i$ contains
at most two qubits).
One can always satisfy conditions of the theorem by adding dummy edges
$(i,j)$ to the interaction graph such that the corresponding terms $H_{i,j}$
in the target Hamiltonian are zero. This would  slightly increase the number
of qubits required for the encoding, see Eq.~(\ref{n}).

Let us  prove the theorem.
Suppose one can choose the local Majorana modes
$\gamma_{i,p}$ such that the following conditions
hold for each vertex $i\in V$ and for
each $1\le p<q\le d(i)$. Here $|O|$ denotes the weight of a Pauli operator $O$.
\be
\label{d3condition}
 |\tilde{B}_i|  \ge 3, \quad 
|\gamma_{i,p}|\ge 2,  \quad |\tilde{B}_i \gamma_{i,p}|\ge 2, \quad
|\tilde{B}_i \gamma_{i,p}\gamma_{i,q}|\ge 2.
\ee
An explicit choice of $\gamma_{i,p}$ satisfying Eq.~(\ref{d3condition}) is shown below.
Assume that $P$  is a  logical operator with weight less than $3$
and show  that  this assumption leads to a contradiction.
Let $V(P)\subseteq V$ be the set of vertices
$i\in V$ such that $P$ acts non-trivially on some qubit of $i$.
By assumption, $|V(P)|~\le~2$.

Suppose first that $V(P)=\{i\}$ is a single vertex
or $V(P)=\{i,j\}$ for some pair of vertices $i\ne j$
such that $(i,j)\notin E$.
Since $P$ commutes with the stabilizers $\tilde{A}(\zeta)$,
it must commute with $\gamma_{i,p} \gamma_{i,q}$ whenever there exists a closed loop $\zeta$
such that $p,q$ are the labels of edges incident to $i$ that belong to $\zeta$.
In the case $V(P)=\{i,j\}$ we additionally require that $\zeta$ does not contain 
the vertex $j$. We claim that such loop $\zeta$ exists 
for all  $1\le p< q\le d(i)$. Indeed, 
let $s=N(i,p)$ and $t=N(i,q)$ be the $p$-th and the $q$-th neighbors of $i$. 
By assumption, $j\notin \{s,t\}$. 
Let $G'$ be the graph obtained from $G$ by removing the vertices
$i,j$ and all edges incident to these vertices.  By assumption,
$G'$ is connected.
Let $\zeta'$ be any path in the graph $G'$ connecting $s$ and $t$.
Complete $\zeta'$ to a  loop $\zeta$ in the graph $G$ by adding the 
edges $(i,s)$ and $(i,t)$. 
By construction,  $\tilde{A}(\zeta)$ acts on the vertex $i$ 
as $\gamma_{i,p} \gamma_{i,q}$ and acts trivially on the vertex $j$.
It follows that 
$P$ commutes with $\gamma_{i,p} \gamma_{i,q}$ for all  $1\le p<q\le d(i)$.
This is possible only if $P \sim \tilde{B}_i$.
This  contradicts to the assumption that $P$ acts on at most two qubits, 
per Eq.~(\ref{d3condition}).

Suppose now that $V(P)=\{i,j\}$ for some pair of vertices  $i\ne j$ such that $(i,j)\in E$. 
We have to consider two cases.

{\em Case~1:} There is a single edge connecting $i$ and $j$.
Suppose $j$ is the $r$-th neighbor, $j=N(i,r)$.
Choose any integers $1\le p<q\le d(i)$ such that 
$r\notin \{p,q\}$. The same argument as above shows that the restriction of $P$ onto the vertex $i$
must commute with $\gamma_{i,p} \gamma_{i,q}$. 
This is possible only if $P$ acts on $i$ as $\gamma_{i,r}$ or $\tilde{B}_i \gamma_{i,r}$.
According to Eq.~(\ref{d3condition}), one can check that 
$\gamma_{i,r}$ and $\tilde{B}_i \gamma_{i,r}$ have weight at least $2$ for all $r$.
Likewise, suppose $i$ is the $q$-th neighbor of $j$, that is, $i=N(j,q)$.
The same argument shows that $P$ acts on $j$ as $\gamma_{j,q}$ or $\tilde{B}_j \gamma_{j,q}$.
Again, these operators have weight at least $2$. Thus $P=P_i P_j$, where 
$P_i$ and $P_j$ have weight at least $2$. Therefore $P$ has weight at least $4$
which is a contradiction. 

{\em Case~2:} There are two edges connecting $i$ and $j$.
Suppose $j$ is the $r$-th and $s$-th neighbor of $i$ for some $r\ne s$.
The same argument as above shows that 
 the restriction of $P$ onto the vertex $i$
must commute with $\gamma_{i,p} \gamma_{i,q}$ for any $1\le p<q\le d(i)$
such that $r,s\notin \{p,q\}$. This is possible only if 
the restriction of $P$ onto the vertex $i$ belongs to the group
generated by  $\gamma_{i,r}$, $\gamma_{i,s}$, and $\tilde{B}_i$.
Likewise, the restriction of $P$ onto the vertex $j$ belongs to the group
generated by $\gamma_{j,t}$, $\gamma_{j,u}$, and $\tilde{B}_j$ for some 
$1\le t<u\le d(j)$.
Using Eq.~(\ref{d3condition}) one can check that 
$P=P_iP_j$ has weight at most $2$ only if 
$P\sim \gamma_{i,r} \gamma_{i,s} \gamma_{j,t} \gamma_{j,u}$. However, such $P$
is proportional to the stabilizer $\tilde{A}(\zeta)$ where
$\zeta$ here is a loop formed by the two edges connecting $i,j$.
This is impossible since $P$ is a logical operator.
To summarize, Theorem~\ref{thm:main} follows from Eq.~(\ref{d3condition}).

Let us show how to satisfy Eq.~(\ref{d3condition}) in the special case of  degree-$6$ graphs.
In this case each vertex $i$ contains $3$ qubits. We shall denote Pauli operators acting on
the qubits located at a vertex $i$ as $(QRT)_i$, where
$Q,R,T\in \{I,X,Y,Z\}$. Choose
\begin{align}
\gamma_{i,1}=(ZXI)_i, \qquad  & \gamma_{i,2}=(ZYI)_i\nonumber \\
\gamma_{i,3}=(IZX)_i , \qquad  & \gamma_{i,4}=(IZY)_i\nonumber \\
\gamma_{i,5}=(XIZ)_i, \qquad & \gamma_{i,6}=(YIZ)_i. \label{degree6Majorana}
\end{align}
Note that $\tilde{B}_j=(ZZZ)_j$.
One can easily check that these operators obey the commutation rules Eq.~(\ref{localMajorana}) and
weight constraints Eq.~(\ref{d3condition}), therefore proving Theorem~\ref{thm:main}
for degree-$6$ graphs. A generalization of Eq.~(\ref{degree6Majorana})
to arbitrary even vertex degree $d(i)\ge 6$ can be found in Appendix~\ref{app:GSEhighdegree}.

\begin{figure}[t!]
	\includegraphics[width=6.5cm]{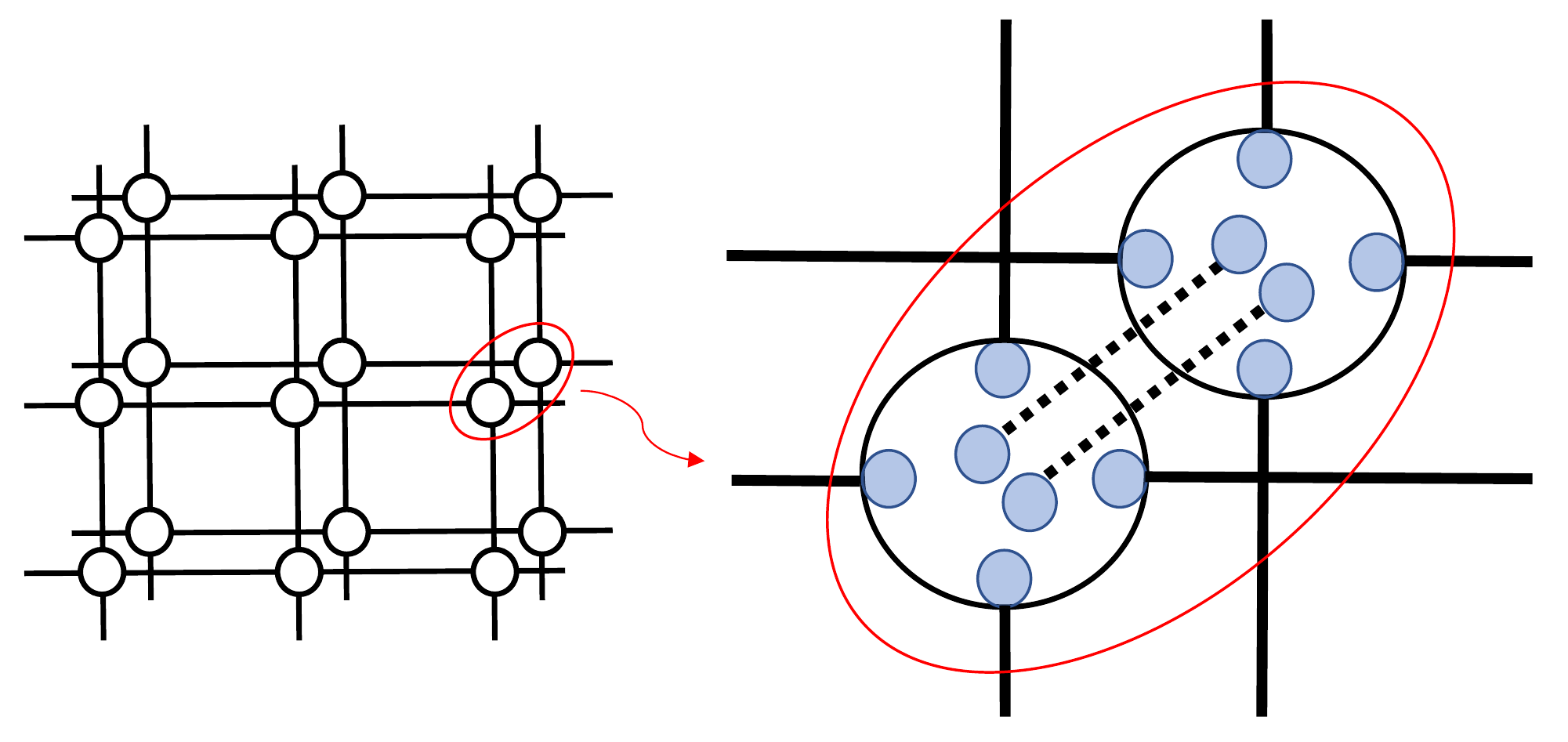}
	\caption{Qubit encoding of the 2D Hubbard model using the GSE.
	{\em Left:} two lattices representing spin-up and spin-down fermi modes.
	{\em Right:} a local view of the interaction graph $G$. 
	Each vertex contains $6$ local Majorana modes ($3$ qubits).
	Dotted lines represent dummy edges introduced to satisfy
	conditions of Theorem~\ref{thm:main}.}
    \label{fig:Hub_GSE}
\end{figure}

{\it Generalized Superfast Encoding for the Hubbard model. }
Let us now show how to simulate the 2D Hubbard model using the GSE. The model Hamiltonian is given as 
\begin{align}
    H=&-t\sum_{(i,j)} \; \sum_{\sigma \in \{\uparrow, \downarrow \} }(a^{\dagger}_{i\sigma}a_{j\sigma}+a^{\dagger}_{j\sigma}a_{i\sigma}) \nonumber \\
    &+\epsilon\sum_i \sum_{\sigma \in \{\uparrow, \downarrow \} }a^{\dagger}_{i\sigma}a_{i\sigma}
    +U\sum_{i}a^{\dagger}_{i\downarrow}a_{i\downarrow}a^{\dagger}_{i\uparrow}a_{i\uparrow}, \label{eq:hub_ham}
\end{align}
where $i,j$ are sites of a square lattice,
$(i,j)$ stand for nearest-neighbor sites,
$\sigma$ is a spin index,
and $t,\epsilon,U$ are some coefficients.
The interaction graph $G$ shown on Fig.~\ref{fig:Hub_GSE}
consists of two copies of the lattice representing spin-up and spin-down fermi modes.
To satisfy  conditions of Theorem~\ref{thm:main}
we have introduced two dummy edges (dotted lines) connecting 
each
respective pair of spin-up and spin-down vertices.  
The resulting graph $G$ 
is $3$-connected and 
has vertex degree $6$. Therefore, by Theorem~\ref{thm:main},
the corresponding GSE encoding corrects any single-qubit error. 
The encoding requires $6s$ qubits, where $s$ is the number of sites in the  lattice
(the number of fermi modes is $m=2s$).
Using Eq.~(\ref{degree6Majorana}) one can check that 
the operators $\tilde{B}_j$, $\tilde{A}_{j,k}$, and $\tilde{A}_{j,k}\tilde{B}_j$ have 
Pauli weight $3$, $4$, and $4$ respectively. 
The two-mode interaction operators $\tilde{B}_j \tilde{B}_k$ have weight $6$.
 We conclude that the simulator
Hamiltonian $\tilde{H}$ for the 2D Hubbard model is a sum of Pauli terms
with weight at most $6$.

{\it Conclusions.}
We  described a GSE that maps a target fermi Hamiltonian
on a graph of degree $d$  into  a qubit simulator Hamiltonian 
with Pauli terms of weight at most $d$
and  corrects all single-qubit errors. 
If one does not insist on the error correction property, 
the Pauli weight of the simulator Hamiltonian
can be reduced to $O(\log{d})$.  
Future research could address  the question of whether GSEs are capable 
of correcting more than one error and whether it is possible
to combine $O(\log{d})$ Pauli weight and the error correction property.

{\it Acknowledgements. }
The authors thank Jay Gambetta, Kristan Temme and  Theodore Yoder for helpful discussions and comments.
SB and AM acknowledge support from the IBM Research
Frontiers Institute. KS and JDW are funded by NSF awards DMR-1747426, 1820747.



\appendix
\newpage
\section{Properties of the loop operators}
\label{app:loop}

Let $\calS$ be an abelian group generated by 
all loop operators $\tilde{A}(\zeta)$
constructed using the Superfast Encoding or its generalized version.
In this section we prove that $-I\notin \calS$ and thus $\calS$ can be viewed
as a stabilizer group of a quantum code.
We show that this code has $m-1$ logical qubits.
To avoid clutter, in this section we shall omit the tilde symbol
in the notations for loop and edge operators.
In other words, in the rest of this section $A(\zeta)$ and $A_{j,k}$ refer
to qubit operators.

Recall that we consider a connected interaction graph $G=(V,E)$.
Define a  {\em path}
of length $s$ as a function 
\[
\zeta \, : \, \{0,1,\ldots,s\} \to V
\]
such that  vertices
$\zeta(j-1)$ and $\zeta(j)$  are nearest neighbors in the graph $G$
for  all $j=1,\ldots,s$.
A path may intersect itself.
We shall use a shorthand notation $|\zeta|=s$ for the length of $\zeta$.
For any path $\zeta$  let $\bar{\zeta}$ be the inverse path such that 
$|\bar{\zeta}|=|\zeta|=s$ and 
$\bar{\zeta}(j)=\zeta(s-j)$ for $0\le j\le s$.  A path is called a {\em loop} if 
$\zeta(s)=\zeta(0)$. 
Finally, suppose $\zeta_i$ are
paths of length $s_i$, where $i=1,2$.   We say that $\zeta_1$ and $\zeta_2$ are 
 {\em composable} if
$\zeta_1(s_1)=\zeta_2(0)$.
Define a composition $\zeta=\zeta_1\circ \zeta_2$ as a path of length $s_1+s_2$
that $\zeta(j)=\zeta_1(j)$ for $0\le j\le s_1$ and $\zeta(j)=\zeta_2(j-s_1)$
for $s_1\le j\le s_1+s_2$. 
For any path $\zeta$ define a path operator
\[
A(\zeta)=i^{s} A_{\zeta(0),\zeta(1)} A_{\zeta(1),\zeta(2)} \cdots A_{\zeta(s-1),\zeta(s)},
\qquad s\equiv |\zeta|.
\]
\begin{lemma}
\label{app:LoopLemma}
Path operators have the following properties:\\
\noindent
(1) For any path $\zeta$ one has $A(\bar{\zeta})A(\zeta)=I$. \\
\noindent
(2)  $A(\zeta_1\circ \zeta_2)=A(\zeta_1)A(\zeta_2)$
for any composable paths.\\
\noindent
(3) If $\zeta$ is a loop then $A(\zeta)$ commutes with all path operators. \\
\noindent
(4)  If $\zeta$ is a loop then $A(\zeta)^\dag =A(\zeta)$.
\end{lemma}
\begin{proof}
We shall use the commutation rules
\begin{align}
A_{j,k}^\dag = A_{j,k}, \quad 
A_{j,k}^2=I, \quad A_{k,j}=-A_{j,k}, \label{appCR1} \\
 A_{j,k} A_{j',k'} = A_{j',k'}A_{j,k} (-1)^{|\{j,k\} \cap \{ j',k'\}|} \label{appCR2}
\end{align}
Let $\zeta$ be a path of length $s$.
By definition,
\begin{align*}
A(\bar{\zeta})A(\zeta)
=&(-1)^s A_{\zeta(s),\zeta(s-1)} \cdots A_{\zeta(1),\zeta(0)} A_{\zeta(0),\zeta(1)}\\ & \cdots A_{\zeta(s-1),\zeta(s)}.
\end{align*}
From Eq.~(\ref{appCR1}) one gets $A_{\zeta(j),\zeta(j-1)} A_{\zeta(j-1),\zeta(j)} = -I$ for all $j$. Thus
$A(\bar{\zeta})A(\zeta)=(-1)^s \cdot (-1)^s I = I$.
Property~(2) follows directly from the definitions.
Suppose $\zeta$ is a loop. Consider an arbitrary edge $(j,k)\in E$.
To prove Property~(3) it suffices to check that $A(\zeta)$ commutes with $A_{j,k}$.
From Eq.~(\ref{appCR2}) one infers that 
$A_{j',k'}$ anti-commutes with $A_{j,k}$ iff $(j',k')$ is an edge
incident to the subset $\{j,k\}$. However, since $\zeta$ is a loop, it contains
even number of edges incident to any subset of vertices. Thus $A(\zeta)$ commutes
with $A_{j,k}$ proving Property~(3).
To prove Property~(4) suppose that $\zeta=\zeta'\circ e$ for some path
$\zeta'$ and some edge $e$ (considered as a path of length one).
The same argument as above shows that $A(e)$ commutes with
$A(\zeta')$. Likewise, if $\zeta=\zeta'\circ e \circ \zeta''$ for some non-empty
paths $\zeta',\zeta''$ and  some edge $e$ then $A(e)$ anti-commutes with
$A(\zeta')$ and $A(\zeta'')$. Repeatedly applying these commutation rules gives
\begin{align*}
A_{\zeta(s-1),\zeta(s)} \cdots A_{\zeta(1),\zeta(2)}A_{\zeta(0),\zeta(1)}=\\
(-1)^s 
A_{\zeta(0),\zeta(1)} A_{\zeta(1),\zeta(2)} \cdots A_{\zeta(s-1),\zeta(s)}
\end{align*}
and proves Property~(4).
\end{proof}
Let $T\subseteq E$ be some fixed spanning tree of $G$
with a fixed root vertex.
For any vertex $j\in V$ let $\omega^j$ be the unique path in $T$ 
 that starts at  the root  and ends at $j$.  
If $\zeta$ is a loop of length $s$ then
\begin{align}
A(\zeta)=&i^{s} A(\bar{\omega}^{\zeta(0)})A(\omega^{\zeta(0)})
A_{\zeta(0),\zeta(1)}\nonumber \\
& A(\bar{\omega}^{\zeta(1)})A(\omega^{\zeta(1)})
\cdots
A(\bar{\omega}^{\zeta(s-1)})A(\omega^{\zeta(s-1)}) \nonumber \\
& A_{\zeta(s-1),\zeta(s)}A(\bar{\omega}^{\zeta(0)})
A(\omega^{\zeta(0)}). \label{app:AAA}
\end{align}
Here we used Property~(1) of Lemma~\ref{app:LoopLemma}
and noted  that $\zeta(s)=\zeta(0)$.
Note that 
\[
iA_{\zeta(j-1),\zeta(j)} = A(e^j), \qquad e^j\equiv [\zeta(j-1),\zeta(j)].
\]
Here $e^j$ is a path of length one that  starts at $\zeta(j-1)$ and ends at $\zeta(j)$.
Regrouping the terms in Eq.~(\ref{app:AAA}) using Property~(2) gives
\[
A(\zeta)=A(\bar{\omega}^{\zeta(0)})
A(\zeta^1) A(\zeta^2)\cdots A(\zeta^s) A(\omega^{\zeta(0)}),
\]
where
\[
\zeta^j=\omega^{\zeta(j-1)} \circ e^j  \circ \bar{\omega}^{\zeta(j)}.
\]
Note that $\zeta^j$ is  a loop that starts and ends at the root.
Finally, Properties~(1,3) give
\begin{align*}
A(\zeta)= &A(\bar{\omega}^{\zeta(0)})A(\omega^{\zeta(0)})A(\zeta^1) A(\zeta^2)\cdots A(\zeta^s) =\\
&
A(\zeta^1) A(\zeta^2)\cdots A(\zeta^s)
\end{align*}
and all operators $A(\zeta^p)$ pairwise commute.
If $e^j$ belongs to the spanning tree $T$ then $\zeta^j$ is a composition of a path 
from the root to one of the vertices $\zeta(j-1)$, $\zeta(j)$ and the inverse path.
Properties~(1,2) imply that $A(\zeta^j)=I$ whenever $e^j\in T$.
We conclude that any loop operator $A(\zeta)$ belongs to the group generated
by the loop operators $A(\zeta^j)$ with $e^j\notin T$. In other words,
\be
\label{appGens}
\calS=\la  A(\zeta^j) \, : \, e^j\notin T\ra.
\ee
We claim that the set of generators of $\calS$ defined in Eq.~(\ref{appGens}) is independent. 
Consider first the Superfast Encoding. Then $A(\zeta^j)$ is the only generator
that anti-commutes with the Pauli $Z$ acting on the edge $e^j$ which implies the
independence property. Consider now the Generalized Superfast Encoding
and some generator  $A(\zeta^j)$.
Let $p$  be the integer such that $\zeta(j)$ is the $p$-th neighbor of $\zeta(j-1)$.
Then $A(\zeta^j)$ is the only generator
that anti-commutes with the local Majorana mode $\gamma_{\zeta(j-1),p}$ which implies the
independence property.
Property~(4) implies that each generator $A(\zeta^j)$ is hermitian. 
Thus $\calS$ is an abelian group that has a set of independent hermitian
Pauli generators. The standard stabilizer formalism then implies that $-I\notin \calS$.
Note that the number of generators in Eq.~(\ref{appGens}) is
$s=|E|-|T|=|E|-|V|+1$. Thus the stabilizer code $\calS$ has
$|E|-s=|V|-1=m-1$ logical qubits.

\section{Local Majorana modes for vertex degree $d\ge 6$}
\label{app:GSEhighdegree}

In this section we show how to choose the
local Majorana modes $\gamma_{i,p}$
that satisfy the error correction condition
Eq.~(\ref{d3condition}) for any even vertex degree
$d(i)\ge 6$.
For example, if $d(i)=8$ or $d(i)=10$ one can
satisfy Eq.~(\ref{d3condition}) by choosing
\begin{align*}
\gamma_{i,1}=ZZXI, \qquad  & \gamma_{i,2}=ZZYI  \\
\gamma_{i,3}=IZZX,\qquad  & \gamma_{i,4}=IZZY \\
\gamma_{i,5}=XIIZ, \qquad & \gamma_{i,6}=YIIZ \\
\gamma_{i,7}=ZXII,\qquad  & \gamma_{i,8}=ZYII \\
\end{align*}
and
\begin{align*}
    \gamma_{i,1}=ZZXII,\qquad & \gamma_{i,2}=ZZYII\\
    \gamma_{i,3}=IZZXI,\qquad & \gamma_{i,4}=IZZYI\\
    \gamma_{i,5}=IIZZX,\qquad & \gamma_{i,6}=IIZZY\\
    \gamma_{i,7}=XIIZZ,\qquad & \gamma_{i,8}=YIIZZ\\
    \gamma_{i,9}=ZXIIZ,\qquad & \gamma_{i,10}=ZYIIZ
\end{align*}
Suppose now that $d(i)/2=2k+1$ for some integer $k$. Set
\begin{align*}
\gamma_{i,1}=\underbrace{Z\cdots Z}_{k} X \underbrace{I\cdots I}_{k}, \qquad &
\gamma_{i,2}=\underbrace{Z\cdots Z}_{k} Y \underbrace{I\cdots I}_{k}\\
\end{align*}
and choose the remaining modes $\gamma_{i,p}$ as $2k$ cyclic shifts of $\gamma_{i,1}$ and $\gamma_{i,2}$.
If $d(i)/2=2k$ for some integer $k$ then set
\begin{align*}
\gamma_{i,1}=\underbrace{Z\cdots Z}_{k} X \underbrace{I\cdots I}_{k-1}, \quad &
\gamma_{i,2}=\underbrace{Z\cdots Z}_{k} Y \underbrace{I\cdots I}_{k-1}\\
\gamma_{i,2k+1}=X\underbrace{I\cdots I}_{k} X \underbrace{Z\cdots Z}_{k-1}, \quad &
\gamma_{i,2k+2}=Y\underbrace{I\cdots I}_{k} Y \underbrace{Z\cdots Z}_{k-1}\\
\end{align*}
and choose  the remaining modes $\gamma_{i,p}$   as $k-1$ cyclic shifts of $\gamma_{i,1}$, $\gamma_{i,2}$,
$\gamma_{i,2k+1}$, $\gamma_{i,2k+2}$.
One can easily check that such local Majorana modes have the desired property Eq.~(\ref{d3condition}).

\section{Lack of error correction in the Superfast Encoding}
\label{app:lackEC}

In this section we prove Lemma~\ref{lemma:lackEC}.
Suppose $G=(V,E)$ is a $d$-regular graph, that is, every vertex has exactly $d$ incident edges.
We assume that 
 edges incident to each vertex $i$   are labeled by integers  $p\in [d]\equiv \{1,2,\ldots,d\}$.
This can be described by a map
\[
\omega \, : \, V \times [d]\to E
\]
such that $\omega(i,1),\ldots,\omega(i,d)$ are the edges
incident to a vertex $i\in V$.
For any $p,q\in [d]$ let
$E_{p,q}\subseteq E$ be the subset of edges labeled by $p,q$, that is,  
\begin{align}
E_{p,q}=&
\{ e=(i,j)\in E\, : \,  & e=\omega(i,p)=\omega(j,q) \quad \mbox{or} \nonumber \\
&&  e=\omega(i,q)=\omega(j,p)\}.\label{Epq}
\end{align}
By definition, $E_{p,q}=E_{q,p}$.

\begin{prop}
\label{prop:P1}
Suppose the Superfast Encoding corrects all single-qubit errors.
Then $E_{1,p}=E_{d,p}=\emptyset$ for $p\in \{1,2\}$ and $p\in \{d-1,d\}$. 
\end{prop}
\begin{proof}
Consider an edge $(i,j)$. If $(i,j)\in E_{1,p}$ with $p=1,2$
then $\tilde{A}_{i,j}$ has weight$\le 2$.
If $(i,j)\in E_{1,p}$ with $p=d,d-1$
then $\tilde{A}_{i,j}\tilde{B}_i$ or $\tilde{A}_{i,j}\tilde{B}_j$   has weight$\le 2$.
If $(i,j)\in E_{d,p}$ with $p=1,2$ then $\tilde{A}_{i,j}\tilde{B}_i$ or $\tilde{A}_{i,j}\tilde{B}_j$   has weight$\le 2$.
If $(i,j)\in E_{d,p}$ with $p=d,d-1$ then $\tilde{A}_{i,j}\tilde{B}_i\tilde{B}_j$ has weight$\le 2$.
\end{proof}
Below we say  that an edge is {\em singular} if it belongs to $E_{1,p}$ or $E_{d,p}$ for some
$p\in [d]$.

\begin{lemma}
\label{lemma:d<=5}
Suppose the interaction graph $G$ has degree $d\le 5$.
Then the Superfast Encoding does not correct all single-qubit errors.
\end{lemma}
\begin{proof}
Assume the contrary and show that this leads to  a contradiction.
Note that every vertex
$i$ has at least two incident singular edges, namely, $\omega(i,1)$ and $\omega(i,d)$.
Thus the total number of singular edges is at least $2|V|$.
Here we noted that $E_{1,1}=E_{1,d}=E_{d,d}=\emptyset$ by Proposition~\ref{prop:P1}.

On the other hand, 
suppose $e=(i,j)$ is a singular edge such that 
$e=\omega(i,1)$ or $e=\omega(i,d)$.
By Proposition~\ref{prop:P1} $e=\omega(j,p)$ where $p\ne 1,2$
and $p\ne d,d-1$. This is only possible if $d=5$ and $p=3$.
Thus the total number of singular edges is at most $|V|$.
This is a contradiction.
\end{proof}

\begin{lemma}
\label{lemma:d=6}
Suppose the interaction graph $G$ has degree $d=6$.
Then the Superfast Encoding does not correct all single-qubit errors.
\end{lemma}
\begin{proof}
Assume the contrary and show that this leads to a contradiction.
The same argument as above shows that  the total number of singular edges is at least $2|V|$.
On the other hand, 
suppose $e=(i,j)$ is a singular edge such that 
$e=\omega(i,1)$ or $e=\omega(i,d)$.
By Proposition~\ref{prop:P1} $e=\omega(j,p)$ where $p\ne 1,2$
and $p\ne d,d-1$. This is only possible if  $p=3$ or $p=4$.
Thus the number of singular edges is at most $2|V|$.
This is only possible if there are exactly $2|V|$ singular edges
and every vertex $i$ has exactly four incident singular edges, namely,
$\omega(i,1)$, $\omega(i,d)$, $\omega(i,3)$, $\omega(i,4)$.

Consider some vertex $i$ and edges
$e=\omega(i,3)$, $f=\omega(i,4)$ incident to $i$.
The above shows that $e$ and $f$ are singular. 
Let $e=(i,j)$ and $f=(i,k)$ for some vertices $j,k\in V$.
Consider two cases.

\noindent
{\em Case~1:} $j\ne k$. 
Then one of the operators
\[
\tilde{A}_{i,j} \tilde{A}_{i,k},  \quad
\tilde{A}_{i,j} \tilde{A}_{i,k} \tilde{B}_j, \quad
\tilde{A}_{i,j} \tilde{A}_{i,k} \tilde{B}_k, \quad
\tilde{A}_{i,j} \tilde{A}_{i,k} \tilde{B}_j \tilde{B}_k
\]
acts non-trivially only on the qubits $e,f$.
Since these are logical operators, we get a contradiction.

\noindent
{\em Case~2:} $j=k$. Then $e=\omega(j,1)$, $f=\omega(j,d)$
or $e=\omega(j,d)$, $f=\omega(j,1)$. In both cases the operator
\[
\tilde{A}_e \tilde{A}_f \tilde{B}_j
\]
acts non-trivially only on the qubits $e,f$.
This is a contradiction since  $\tilde{A}_e \tilde{A}_f \tilde{B}_j$ is a logical operator.
\end{proof}

\section{2D Hubbard Model}
\label{app:Hubbard}
Here we derive an encoding for the Hubbard model using the original Superfast Algorithm that incorporates single-qubit error correction. Note that the distance of the graph necessary to do this is $d=8$.
We use the same lattice structure as the one used in the main texdt for the GSE, e.g. two square lattices of opposite spins connected by vertical edges. The Hamiltonians for the two lattices are given in Eq.~\eqref{eq:hub_ham} in the main text. 

 \begin{figure}[h]
	\includegraphics[width=9cm]{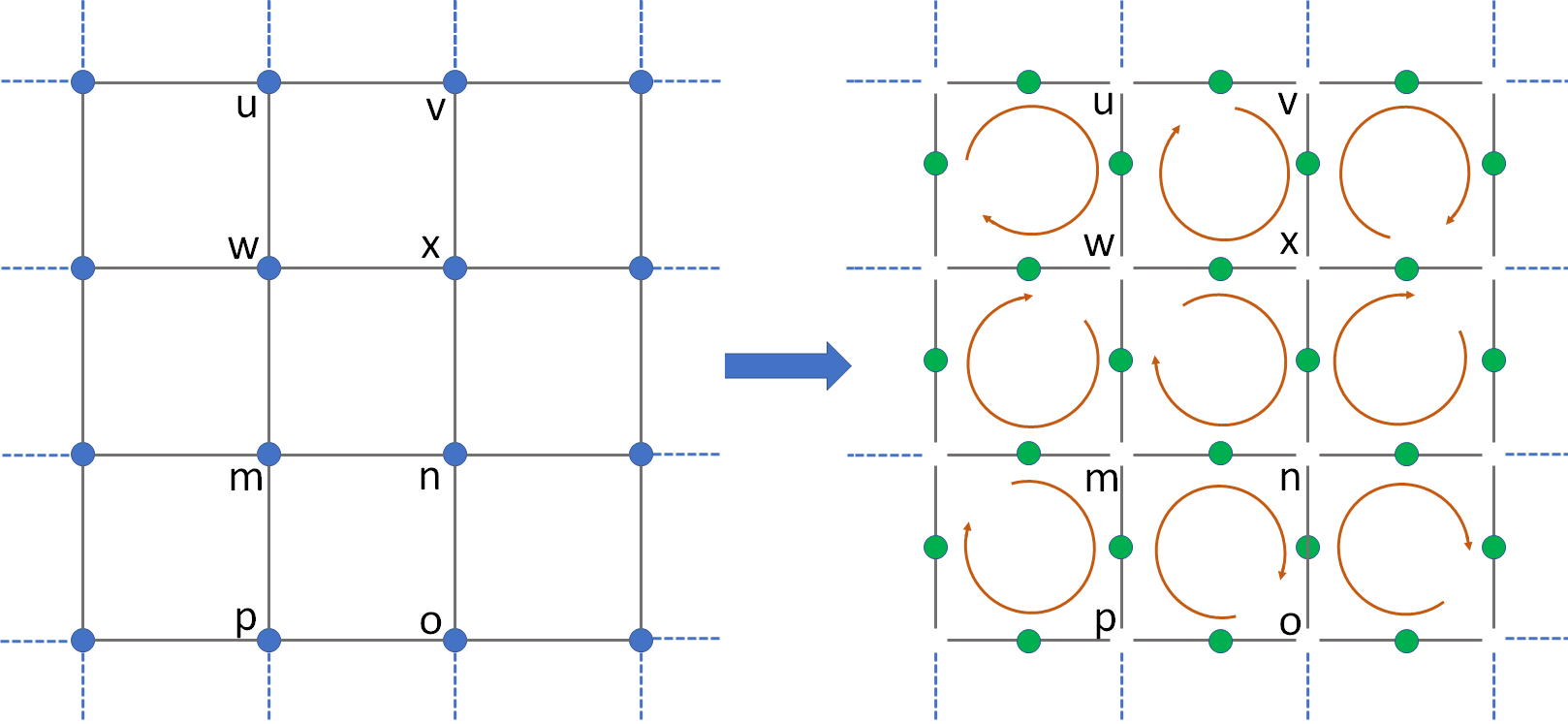}
	\caption{Encoding of the Hubbard model using the Superfast Encoding.}
    \label{fig:Hubbard_mt}
\end{figure}

We transform the creation and annihilation operators of each spin square lattice to edge operators 
The graph $G$ describing one of the two spin lattices for the 2D Hubbard model is shown on Fig.~\ref{fig:Hubbard_mt}. Fermi modes (blue dots) lie on the vertices and the edges represent hopping operators. Qubits of the Superfast Encoding live on edges of the lattice (green dots). The relevant stabilizer operators correspond to the elementary  loops (plaquettes).
For example, the loop $\zeta=(u,v,x,w)$ shown on Fig.~\ref{fig:Hubbard_mt} 
gives rise to a stabilizer
\[
  \tilde{A}(\zeta)=\tilde{A}_{u,v}\tilde{A}_{v,x}\tilde{A}_{x,w}\tilde{A}_{w,u}
   =X_{uv} X_{vx} X_{xw} X_{wu}\cdots 
\]
where the dots represent a product of Pauli $Z$ on some edges
incident to $u,v,x,w$ that 
depend on the chosen edge ordering. Let $\calS$ be the stabilizer group
generated by all loop operators.

There are three distinct terms that appear in the Hubbard model, excitation term, number operator term, and the Coloumb operator term. Based on the expressions found in \cite{Setia2017}, we know the edge operator representation of all the three terms that appear in the Hubbard model. Therefore, for $H_{\uparrow}$ we get:
$$H_{\uparrow}=-t\sum_{ij}\frac{-i}{2}(A_{ij\uparrow}B_{j\uparrow}+B_{i\uparrow}A_{ij\uparrow})+\epsilon\sum_{i}(\frac{1-B_{i\uparrow}}{2})$$

The spin-density interaction terms are given by:
$$U\sum_{i}n_{i \uparrow}n_{i\downarrow}=U\sum_{ij}(\frac{1-B_{i\uparrow}}{2})(\frac{1-B_{i\downarrow}}{2})$$

\begin{figure}[h]
	\includegraphics[width=6cm]{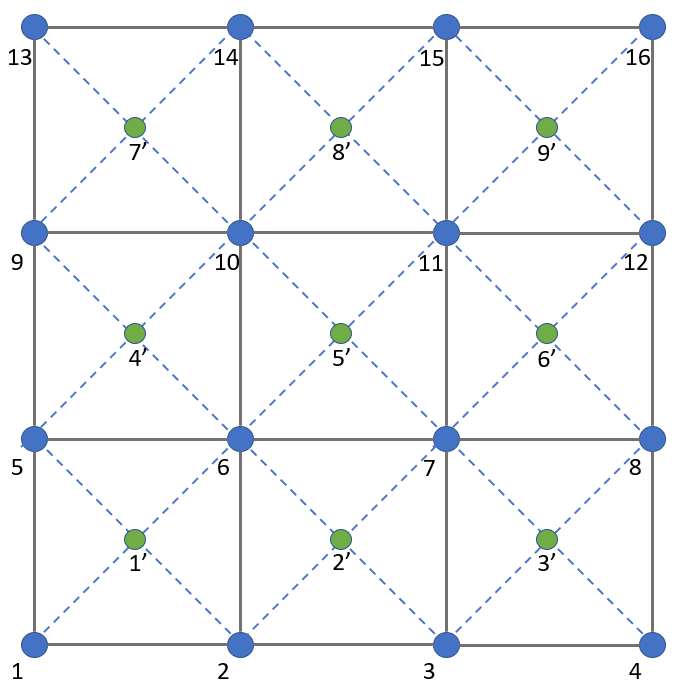}
	\caption{$4\cross4$ Hubbard model lattice with auxiliary modes inserted. The blue vertices are the modes present due to the original problem and the green vertices correspond to the auxiliary modes introduces for error correction. The vertical and horizontal solid line edges correspond to the original fermionic problem Hamiltonian. The dashed lines correspond to the edges introduced due to the auxiliary mode.}
    \label{fig:bksf_aux_modes}
\end{figure}

As discussed in Section~\ref{app:lackEC}, the set of stabilizers available proves to be insufficient
to correct all the single-qubit errors in the Hubbard model, if we do not introduce auxiliary ancillary modes. 
These auxiliary modes contribute to edges in the graph but do not have fermionic terms appearing in the target  Hamiltonian. For each plaquette  in the original lattice we introduce an auxiliary mode at its center which `interacts' with all the vertices of the plaquette, see  
Fig.~\ref{fig:bksf_aux_modes}. We get four extra edges per one auxiliary mode in the model, which give us four smaller independent stabilizer loops. We can then the $B_i$ vertex operators at the auxiliary mode as stabilizers. 

To prove that the code corrects all single-qubit Pauli errors it suffices to show that
each single-qubit error has a unique syndrome.
From Eqs.~(\ref{eq:B_i},\ref{eq:A_ij}) it is easy to see that the ordering of the edges will affect the analytical expressions of the stabilizers. This in turn affects whether it is possible to get unique syndromes for all the single qubit errors or not. The ordering that we used is given on a unit cell in Fig.~\ref{fig:bksf_aux_modes}. The fermionic modes in the original problem are represented with numbers without dashes, while auxiliary modes are numbered with a dash. We use the ordering $1'<2'<3'..<1<2<3..$ . The bottom most row is numbered from left to right and then the numbering continues for the rows above it. So, the mode numbers increase from left to right and from bottom to top. In
 Fig.~\ref{fig:bksf_aux_modes}, for any mode, the mode left and above it will have a higher index. Due to the ordering choice, we can prove single-qubit error correction for a unit cell in terms of stabilizers around it. Indeed, it is easy to check that all single-qubit errors in the unit cell '6-7-11-10'
have distinct syndromes. 
Note that this encoding requires $12s$ qubits, where $s$ is the number of sites in the original lattice.

\section{Fermionic operators in terms of edge operators} \label{app:ferm_op_table}
Operators $A_{ij}$ and $B_i$ can generate the algebra of all even fermionic operators. This is the case for a generic quantum chemistry Hamiltonian
\be
H=\sum_{ij}h_{ij}a_{i}^{\dagger}a_j+\sum_{ijkl}h_{ijkl}a_i^{\dagger}a_{j}^{\dagger}a_ka_l.
\ee
Here, there are five different types fermionic operators namely, number operator $a_{i}^{\dagger}a_i$, excitation operator $a_i^{\dagger}a_j$, number excitation operator $a_i^{\dagger}a_j^{\dagger}a_ja_k$, Coloumb operator $a_i^{\dagger}a_j^{\dagger}a_ja_i$, double excitation operator $a_i^{\dagger}a_j^{\dagger}a_ka_l$. Their expression in terms of edge operators are given in the Table \ref{Table:operators}. We have also included derivation of a superconductivity interaction of the form
$a^{\dagger}_{i}a_{j}^{\dagger}+a_ia_{j}$.

\begin{center}
\begin{table}[h]
	\caption{Edge operator representation for even fermionic operators}
	\label{Table:operators}
	\begin{tabular}{c c}
		\hline
		\hline
		\textbf{Second quantized form} & \textbf{Edge Operator Representation} \\
		\hline
	    $a^{\dagger}_{i}a_{i}$&$(1-B_{i})/2$\\
	    $a^{\dagger}_{i}a^{\dagger}_{j}a_{j}a_{i}$&$\left(1-B_{i}\right)\left(1-B_{j}\right)/4$\\
		$(a^{\dagger}_{i}a_{j}+a^{\dagger}_ja_{i})$& $-\ii (A_{ij}B_{j}+B_{i}A_{ij})/2$\\
		$(a^{\dagger}_{i}a^{\dagger}_{j}a_{j}a_{k}+a^{\dagger}_{k}a^{\dagger}_{j}a_{j}a_{i})$&$-\ii (A_{ik}B_{k}+B_{i}A_{ik})(1-B_{j})/4$\\
		$(a^{\dagger}_{i}a_{j}^{\dagger}+a_ia_{j})$& $-\ii (A_{ij}B_{j}-B_{i}A_{ij})/2$\\
		$(a^{\dagger}_{i}a^{\dagger}_{j}a_{k}a_{l}+a^{\dagger}_{l}a^{\dagger}_{k}a_{j}a_{i})$&$A_{ij}A_{kl}(-1-B_iB_j+B_iB_k+B_iB_l+$\\
	&$B_jB_k+B_jB_l-B_kB_l-B_iB_jB_kB_l)/8$\\
		\hline
		\hline
	\end{tabular}
\end{table}
\end{center}

\end{document}